\documentclass[11pt,reqno]{amsart}
\usepackage{graphicx}
\usepackage{stmaryrd}
\usepackage{amssymb}
\usepackage{amsthm}
\usepackage{dsfont}
\usepackage{hyperref}
\usepackage{mathrsfs}
\usepackage{stmaryrd}
\usepackage[lite,initials,msc-links]{amsrefs}
\usepackage{amsmath}
\usepackage{centernot}
\usepackage{mathtools}
\usepackage{xcolor}

\newtheorem{theorem}{Theorem}
\newtheorem{lemma}[theorem]{Lemma}
\newtheorem{proposition}[theorem]{Proposition}
\newtheorem{corollary}[theorem]{Corollary}
\numberwithin{equation}{section}
\theoremstyle{definition}
\newtheorem{example}{Example}[section]

\usepackage{chngcntr}
\counterwithout{equation}{section}

\theoremstyle{definition}
\newtheorem*{definition}{Definition}
\theoremstyle{remark}
\newtheorem{remark}{Remark} 
\allowdisplaybreaks

\newcommand{\IP}{\mathbf{P}}

\newcommand{\con}[1]{ \xleftrightarrow{#1}}
\newcommand{\ncon}[1]{ \centernot{\xleftrightarrow{#1}}}

\newcommand{\n}{\mathbf{n}}

\newcommand{\pev}{{ \mathcal{P}}}

\newcommand{\nod}{\mathcal{N}}

\newcommand{\at}{}

\DeclareMathOperator{\pf}{\textnormal{pf}}

\newcommand{\red}{\mathcal R}
\newcommand{\blue}{\mathcal B}
\newcommand{\sink}{\mathcal S_-}
\newcommand{\origin}{\mathcal S_+}
\newcommand{\source}{\mathcal S}

\makeatletter
\newcommand\connect[2][]{%
  \ext@arrow 9999{\longleftrightarrowfill@}{#1}{#2}}
\newcommand\longleftrightarrowfill@{%
  \arrowfill@\leftarrow\relbar\rightarrow}
\makeatother

\title[On boundary correlations in planar Ashkin--Teller models]{On boundary correlations in planar Ashkin--Teller models}

\author{Marcin Lis}{
\address{Faculty of Mathematics\\ University of Vienna \\
Oskar-Morgenstern-Platz 1\\
1090 Wien}
\email{marcin.lis@univie.ac.at}}
 
\date{}

\begin{document}

\maketitle

\begin{abstract}
We generalize the switching lemma of Griffiths, Hurst and Sherman to the random current representation of the Ashkin--Teller model.
We then use it together with properties of two-dimensional topology to derive linear relations for multi-point boundary spin correlations and bulk order-disorder correlations in planar models.

We also show that the same linear relations are satisfied by products of Pfaffians.
As a result a clear picture arises in the noninteracting case of two independent Ising models
where multi-point correlation functions are given by Pfaffians and determinants of their respective two-point functions.
This gives a unified treatment of both the classical Pfaffian identities and recent total positivity inequalities for boundary spin correlations in the planar Ising model.

We also derive the Simon and Gaussian inequalities for general Ashkin--Teller models with negative four-body coupling constants.

\end{abstract}

\section{Introduction}
It has been well known since the work of Groeneveld, Boel and Kasteleyn~\cite{GBK} that in the Ising model the multi-point correlations of spins lying on the boundary of a planar graph
are given by Pfaffians of the respective two-point correlations. 
This can be seen as the Wick's rule for expectation values of products of noninteracting {Majorana} fermions, and is one of the many manifestations of the fermionic structure underlying the planar Ising model, going back to the work of Onsager and Kaufman~\cite{Onsager,Kaufman}, and Kadanoff and Ceva~\cite{KC}. 
Recently it was noticed by the author that certain matrices of such boundary two-point functions are totally positive~\cite{LisT}, i.e., determinants 
of all their minors are positive.
This was later extended by Galashin and Pylyavskyy to a deep and in a well defined sense bijective relation between planar Ising models and the positive orthogonal Grassmannian~\cite{GaPy}, which in turn was introduced in the study of scattering amplitudes of ABJM theory~\cite{HW}.
In this article we present a unified framework from which both the classical Pfaffian identities and total positivity inequalities of Ising boundary correlations are naturally concluded. Unlike the previous 
proofs of total positivity, our approach does not use mappings to other models like alternating flows~\cite{LisT} or the dimer model~\cite{GaPy}. Instead we use switching lemmas for random currents.
The idea of using switching identities to prove Pfaffian relations of Ising correlations (though to some extent implicitly present in the work of Boel and Kasteleyn~\cite{BK}) originated in the recent article of Aizenman, Duminil-Copin, Tassion and Warzel~\cite{ADTW} which was an inspiration for the present work. We note that related methods were also applied by Aizenman, Valc\'{a}zar and Warzel to the dimer model~\cite{AVW}.

More generally we establish linear identities and inequalities that are satisfied by boundary correlations of planar Ashkin--Teller models~\cite{AT}, i.e., two Ising-like spin configurations $\sigma$ and $\tilde \sigma$ coupled by 
a Hamiltonian with a four-body interaction.
To this end we first define a random current representation of the model, and establish a switching identity for the correlations of $\sigma$ and $\tilde \sigma$ which is a generalization of the classical switching lemma of Griffiths, Hurst, and Sherman~\cite{GHS}
for two independent Ising models (the case of vanishing four body interactions).  
We note that a similar idea, though expressed in a different language, appeared already in the work of Chayes and Shtengel~\cite{ChSh}.
Subsequently we show a new switching identity for the correlations of the $\{-1,0,1\}$-valued spins $\varphi=(\sigma+\tilde \sigma)/2$ and $\tilde \varphi=(\sigma-\tilde \sigma)/2$. This yields a set of linear inequalities for the correlations of $\sigma$ and~$\tilde\sigma$, that in the case of independent Ising models were originally established by Kasteleyn and Boel~\cite{KB}.
Moreover, the desired linear identities follow from the crucial observation that the correlations of $\varphi$ and $\tilde \varphi$ may be forced to vanish by properly choosing the order of spin insertions on the 
boundary of a planar graph.

We also show that the same relations are satisfied by products of Pfaffians. Since the correlations of $\sigma$ and $\tilde \sigma$ factorize in the noninteracting case, 
a unified picture arises:
The boundary correlations of $\sigma$ are given by {Pfaffians} of their respective two-point functions, whereas the mixed correlations of $\varphi$ and $\tilde \varphi$ 
are given by analogous {determinants}. 
The latter may be thought of as an instance of the fermionic Wick's rule for expectation values of products of noninteracting {Dirac} fermions.
This picture should be compared with the bosonic Wick's rule which states that higher moments of real Gaussian fields are {hafnians} and those of a complexified pair of independent Gaussian fields are {permanents} of their second moments.
Interestingly, in our setting total positivity of two-point boundary correlations turns out to be intrinsically related to the first Griffiths inequality for the spins $\varphi$ and $\tilde \varphi$.

Furthermore, using slightly more involved topological considerations we establish linear relations for the correlations of Kadanoff--Ceva fermions~\cite{KC,Fisher,CCK,ADTW}.
Again, in the non-interacting case this yields Pfaffian relations (that were obtained using random currents for the first time in~\cite{ADTW}) and new determinantal identities.

Our considerations also shed light on the classical but nonetheless intriguing phenomenon: a celebrated result of Fisher~\cite{Fisher} says that a single instance of a planar Ising model has a 
representation in terms of a {nonbipartite} dimer model, and another well known result of Dub\'{e}dat~\cite{dubedat} (see also \cite{bdt,DCL})
states that two independent copies of the Ising model can be realized as a {bipartite} dimer model. This matches the picture presented in this paper
as correlations of monomer insertions in nonbipartite dimers are given by Pfaffians, whereas those of bipartite dimers are determinants~\cite{Kasteleyn}. 

Finally, in a departure from the planar setup we use the switching lemma to derive the Simon and Gaussian inequality for the Ashkin--Teller model with nonpositive four body interaction on general graphs.
The Simon inequality classically implies a certain type of sharpness of the phase transition, namely that finite susceptibility implies exponential decay of correlations. 

This article is organized as follows: In Sect.~\ref{sec:AT} we define the Ashkin--Teller model and its random current representation, and prove switching lemmas for the correlations of $\sigma,\tilde \sigma$ and $\varphi,\tilde \varphi$ for general, not necessarily planar, graphs. In Sect.~\ref{sec:planar}, we consider the planar setting and characterize all correlations of $\varphi, \tilde \varphi$ that vanish for topological reasons. 
In Sect.~\ref{sec:pfdet}, we show that products of Pfaffians satisfy the same equations as the correlations of $\sigma, \tilde \sigma$ from Sect.~\ref{sec:planar}.
This implies that in the non-interacting case the correlations of $\sigma,\tilde \sigma$ are given by Pfaffians and those of $\varphi,\tilde \varphi$ are given by determinants.
In particular total positivity of certain matrices of two-point functions is recovered. 
In Sect.~\ref{sec:disorders} we study order-disorder correlations, and using slightly more complicated topological arguments (involving the notion of double covers) we obtain linear relations
for the correlations of Kadanoff--Ceva fermions.
Finally, in Sect.~\ref{sec:simon} we show the Simon and Gaussian inequality for general Ashkin--Teller models with nonpositive four body interaction. 
\section{The Ashkin--Teller model and the switching lemma} \label{sec:AT}
Let $G=(V,E)$ be a finite graph.
The Ashkin-Teller model~\cite{AT} (with free boundary conditions) is a probability measure on pairs of spin configurations $(\sigma,\tilde \sigma)\in \{ -1,1\}^V\times \{ -1,1\}^V  $
given by
\begin{align}\label{def:AT}
\IP_{\at} (\sigma,\tilde \sigma)=\frac1{\mathcal Z}\prod_{uv\in E} \exp\big(J_{uv} (\sigma_u \sigma_v+\tilde \sigma_u\tilde \sigma_v)+U_{uv} (\sigma_u \sigma_v\tilde \sigma_u\tilde \sigma_v+1)\big),
\end{align}
where $\mathcal Z$ is the partition function, and $J_{uv}$ and $U_{uv}$ are coupling constants. The constant $U_{uv}$ is added to the Hamiltonian for convenience as it does not change the probability measure.
Note that the case $U= 0$ is equivalent to two independent Ising models.
For $A,B\subseteq V$, let  $\sigma_A=\prod_{v\in A}\sigma_v$, $\tilde \sigma_B=\prod_{v\in B}\tilde\sigma_v$. We use the convention that $\sigma_{\emptyset}=\tilde \sigma_{\emptyset}=1$. Since the spins are $\pm 1$-valued, the law of the model is completely described by all spin correlation functions of the form
\[
 \langle \sigma_A \tilde \sigma_B \rangle_{\at} =  \sum_{\sigma,\tilde \sigma \in  \{ -1,1\}^V } \sigma_A \tilde \sigma_B  \IP_{\at} (\sigma,\tilde\sigma).
 \]
 
We want to study the random current representation of such correlations.
For the purpose of this article we follow~\cite{LisT,DCL} and use a different than the classical~\cite{GHS,aizenman} but equivalent definition of currents (see~\cite{DC} for an account of random currents in the Ising model). 
To this end, for a set of edges~$\eta$, let $\delta(\eta)$ be the set of vertices of odd degree in the graph $(V,\eta)$.
We say that a pair $ \n=(\omega,\eta)$, where $\omega,\eta\subseteq E$,
is a \emph{current} with \emph{sources} $A$ if $\eta\subseteq \omega$ and $\delta(\eta)=A$ (one can think of $\omega\setminus\eta$ as the edges with nonzero even values in the classical definition of a current, and 
of $\eta$ as the odd valued edges). 
We write $\Omega_A$ for the set of all currents with sources~$A$ when $|A|$ is even, and we set $\Omega_A=\emptyset$ otherwise.
We define the Ashkin--Teller weight of a current by
\begin{align}\label{eq:currweight}
w_{\at}(\n)=2^{k(\omega)} \prod_{e\in \eta} x_{e} \prod_{e\in\omega\setminus \eta} y_{e} ,
\end{align}
where $k(\omega)$ is the number of connected components of the graph $(V,\omega)$ including isolated vertices, and where the weights $x_e=x_e(U_e,J_e)$ and $y_e=y_e(U_e,J_e)$ are given by
\begin{align} \label{eq:CD}
x_e=e^{2U_e} \sinh(2J_e)\qquad \textnormal{and} \qquad  y_{e}=e^{2U_e}\cosh(2J_e)-1 .
\end{align}
We will write $Z_{\emptyset}=\sum_{\n\in \Omega_{\emptyset}} w_{\at}(\n)$ for the partition function of sourceless currents. 
Note that if 
\begin{align} \label{eq:nonnegative}
J\geq 0 \quad \text{and} \quad \cosh(2J_e)\geq e^{-2U_e},
\end{align} then the weights are nonnegative.

Our first result is a generalization of the switching lemma of Griffiths, Hurst, and Sherman~\cite{GHS} to the Ashkin--Teller model.
A related idea appeared already in the work of Chayes and Shtengel~\cite{ChSh}.
\begin{proposition}[Switching lemma for $\sigma$ and $\tilde \sigma$]\label{prop:switch}
Let $A,B\subseteq V$. Then for all coupling constants $J$ and $U$,
\begin{align} \label{eq:Switch}
\langle  \sigma_A \tilde \sigma_{B}\rangle_{\at} = \frac1{Z_{\emptyset}} \sum_{\n=(\omega,\eta)\in\Omega_{A\triangle B}   } w_{\at}(\n)\mathbf 1\{\omega \in \frak F_B\},
\end{align}
where $\triangle$ denotes the symmetric difference, and where $\omega \in\frak F_B$ if and only if each connected component of $(V,\omega)$ contains 
an \emph{even} number of vertices from~$B$.
\end{proposition} 

 \begin{proof}
Consider the spin $\tau_v=\sigma_v\tilde \sigma_v$. Writing $\delta_{\tau_u\tau_v} $ for the Kronecker delta, we have
\begin{align*}
\exp\big(J_{uv} (\sigma_u \sigma_v+\tilde \sigma_u\tilde \sigma_v)+& U_{uv}( \sigma_u \sigma_v\tilde \sigma_u\tilde \sigma_v+1)\big) =1+ \delta_{\tau_u\tau_v}(x_{uv} \sigma_u\sigma_v+y_{uv}  ),
\end{align*}
where we use that the spins are $\pm 1$-valued.
We can now expand these factors and write
\begin{align*}
\mathcal Z \langle \sigma_A \tilde \sigma_{B}\rangle_{\at} &=\mathcal Z\langle \sigma_{A\triangle B}  \tau_{B}\rangle_{\at}\\
&= \sum_{\sigma,\tau} \sigma_{ A\triangle B} \tau_B\prod_{uv\in E}\big(1+ \delta_{\tau_u\tau_v}\big(x_{uv} \sigma_u\sigma_v+y_{uv}  ) \big)\\
&= \sum_{\sigma,\tau} \sigma_{ A\triangle B} \tau_B \sum_{\omega\subseteq E} \prod_{uv\in \omega} \delta_{\tau_u\tau_v}(x_{uv} \sigma_u\sigma_v+y_{uv}  )\\
&= \sum_{\sigma}  \sum_{\omega\subseteq E}\mathbf 1\{\omega\in \frak F_B\}2^{k(\omega)} \sigma_{ A\triangle B} \prod_{uv\in \omega} ( x_{uv} \sigma_u\sigma_v+y_{uv}  )\\
&=\sum_{\sigma}   \sum_{\omega\subseteq E} \mathbf 1\{\omega\in \frak F_B\}2^{k(\omega)} \sum_{\eta \subseteq \omega} \sigma_{ A\triangle B\triangle \delta(\eta)} 
\prod_{e\in  \eta}x_e \prod_{e\in \omega \setminus \eta}y_e  \\
&={2^{|V|}}\sum_{\omega\subseteq E} \mathop{\sum_{\eta \subseteq \omega}}_{\delta(\eta)=A\triangle B}  \mathbf 1\{\omega\in \frak F_B\}
2^{k(\omega)}\prod_{e\in  \eta}x_{e} \prod_{e\in \omega \setminus \eta}y_{e}  \\
&= {2^{|V|}}\sum_{\n\in \Omega_{A \triangle B}} w_{\at}(\n) \mathbf 1\{\omega\in \frak F_B\}. 
\end{align*}
Indeed, for a fixed $\omega$, summing out the $\tau$ variable results in the factor $2^{k(\omega)}$ as $\tau$ has to be constant on the clusters of $\omega$. Moreover, the factor $\mathbf 1\{\omega\in \frak F_B\}$ appears since if there exists a cluster of $\omega$ with an odd number of vertices in~$B$, then $\tau_B$ changes sign when one flips the spin of this cluster, and hence the sum is zero.
Furthermore, summing over the $\sigma$ variable results in restricting the sum to those $\eta$ with sources at $A\triangle B$, as otherwise $\sigma_{ A\triangle B\triangle \delta(\eta)}$
changes sign as one flips the spin of a vertex in $A\triangle B\triangle \delta(\eta)$.
Taking $A=B=\emptyset$ gives $\mathcal Z=2^{|V|}Z_{\emptyset}$ which completes the proof.
 \end{proof}
 
 \begin{remark}
Based on the above computation one can think of currents as a simultaneous FK representation for the spins $\tau=\sigma\tilde\sigma$, and a high-temperature expansion for $\sigma$.
This mixture of representations is exactly what makes the switching lemma work.
\end{remark}
 
We note that for $U=0$ we recover (a special case of) the classical switching lemma for the Ising model.
We also note that for $(\omega,\eta)\in\Omega_{A\triangle B}$, we have $\omega\in\mathfrak F_{A\triangle B}$, and hence we could as well replace $\mathfrak F_B$ by $\mathfrak F_A$ in the statement of the proposition as $\mathfrak F_{A\triangle B}\cap \mathfrak F_{A}= \mathfrak F_{A\triangle B}\cap \mathfrak F_{B}$.
Moreover, since $\mathfrak F_B=\emptyset$ if $|B|$ is odd, the above correlation functions vanish unless $|A|$ and $|B|$ are even. This can also be seen directly from the definition of the Ashkin--Teller model.

A few words should be also devoted to the name of this result (which is arguably more fitting in the original formulation involving a counting of subgraphs of a given multigraph, see e.g.~\cite{aizenman,ADCS}).
Indeed, as long as $A\triangle B$ is fixed, $B$ enters the expression on the right-hand side of~\eqref{eq:Switch} only through the indicator function that restricts the sum over currents in $\Omega_{A\triangle B}$ 
to those for which $\omega \in \mathfrak {F}_B$. This means that switching from $\tilde \sigma$ to $\sigma$ amounts to only changing this indicator function, and hence the name is justified. For instance, a direct consequence is the following inequality.
\begin{corollary} \label{cor:easy} If the coupling constants are as in~\eqref{eq:nonnegative}, then
\[
\langle \sigma_{A }\sigma_B \rangle \geq \langle  \sigma_A \tilde \sigma_B\rangle.
\]
\end{corollary}

We now turn our attention to correlation functions of the variables $\varphi,\tilde \varphi \in \{-1,0,1\}$ defined by
\[
\varphi_v =\frac{\sigma_v+\tilde \sigma_v}{2} \qquad \text{ and } \qquad \tilde \varphi_v=\frac{\sigma_v-\tilde \sigma_v}{2}.
\]
We write $\varphi_{A}=\prod_{v\in A}\varphi_v$, and $ \tilde \varphi_{A}=\prod_{v\in A}\tilde \varphi_v$ as before.
Note that $\varphi^3=\varphi$ and $\tilde \varphi^3 =\varphi$, and hence it is enough to look at correlations of order at most two.
Also note that $\varphi\tilde \varphi=0$, and we only need to consider insertions of $\varphi$ and $\tilde \varphi$ at disjoint sets of vertices.
It turns out that correlations of $\varphi$ and $\tilde \varphi$ have a natural representation in terms of currents as well, but this time the corresponding indicator function
restricting the sum over currents has a more explicit topological meaning.

\begin{proposition}[Switching lemma for $\varphi$ and $\tilde \varphi$]\label{lem:RB}
Let $A=A_1 \cup A_2, B= B_1 \cup B_2\subseteq V$ be such that $A_1,A_2, B_1,B_2$ are pairwise disjoint. Then for all coupling constants $J$ and $U$,
\begin{align*}
\langle \varphi_{A_1} \varphi^2_{A_2} \tilde \varphi_{B_1} \tilde \varphi^2_{B_2}\rangle = \frac{1}{Z_{\emptyset}} \sum_{\n=(\omega,\eta)\in \Omega_{A_1 \cup B_1} }w_{\at}(\n)2^{-k_{A\cup B}(\omega)}\mathbf{1}
\{ A \ncon{\omega} B\},
\end{align*}
where $k_{A\cup B}(\omega)$ is the number of connected components of $\omega$ intersecting~$A\cup B$, and where $A \ncon{\omega} B$ means that no vertex in $A$ is connected to a vertex in $B$ via a path of edges in $\omega$.
\end{proposition}
\begin{proof}
For a finite set $X$, let $\pev(X)$ be the set of subsets of $X$ of \emph{even} cardinality. 
We have 
\[
\varphi_v^2=\frac{1+\sigma_v\tilde \sigma_v}2= \sigma_v \varphi_v, \qquad \text{and} \qquad \tilde \varphi_v^2=\frac{1-\sigma_v\tilde \sigma_v}2= \sigma_v \tilde \varphi_v.
\]
Therefore 
\begin{align*}
 \langle \varphi_{A_1} \varphi^2_{A_2} \tilde \varphi_{B_1} \tilde \varphi^2_{B_2}\rangle & = \langle \sigma_{A_2\cup B_2}\varphi_{A} \tilde \varphi_{B} \rangle \\
&=2^{-|A\cup B|}\sum_{S \in\pev( A\cup B)} (-1)^{|S\cap B|} \langle  \sigma_{(A_1\cup B_1)\triangle S} \tilde\sigma_S \rangle.
\end{align*}
By the switching lemma applied to each term on the right-hand side, we obtain $1/Z_{\emptyset}$ times
\begin{align*}
& 2^{-|A\cup B|} \sum_{\n=(\omega,\eta)\in \Omega_{A_1\cup B_1}   }w_{\at}(\n)\Big( \sum_{S \in\pev (A\cup B)} (-1)^{|S\cap B|}\mathbf 1\{\omega\in \frak F_S\}\Big).
\end{align*}
Hence it is enough to show that the second sum is equal to
\begin{align*}
 2^{| A\cup B|-k_{A\cup B}(\omega)}\mathbf{1}\{ A \ncon{\omega} B\}.
\end{align*}
To this end note that if $A \ncon{\omega} B$, then the sign in the sum is constant and equal to one.
This accounts for the factor $2^{| A\cup B|-k_{A\cup B}(\omega)}$ which is the number of sets $S\in\pev( A\cup B)$ such that $\omega\in \frak F_S$. Otherwise, take $u\in A$ and $v \in B$ in the same connected component of $\omega$. Then $S\mapsto S\triangle\{u,v\}$ is a sign-reversing involution on sets satisfying $S\in\pev( A\cup B)$ and $\omega\in \frak F_S$. As a result, the above sum is zero.
\end{proof}

This identity is valid for general, not necessarily planar, graphs, and to the best of our knowledge, is new also in the noninteracting case.
Note that if $\{ A \ncon{\omega} B\}$, then necessarily $\omega \in \frak F_{A_1}\cap \frak F_{B_1}$. In particular the correlations above are zero if $|A_1|$ or $|B_1|$ is odd. Moreover they are nonnegative (satisfy the first Griffiths inequality) for coupling constants as in \eqref{eq:nonnegative}. 
When expanded into the correlations of $\sigma$ and $\tilde \sigma$, this yields a collection of linear inequalities that were first obtained for $U=0$ by Kasteleyn and Boel~\cite[Eq. 31]{KB} as a special case of what they call maximal $\Lambda$-inequalities. However, no combinatorial interpretation of the inequalities was given in~\cite{KB}.

A crucial observation is that due to the disconnection condition $A \ncon{\omega} B$ these correlations vanish for topological reasons if the sets $A$ and $B$ are properly chosen.
For example, let $A=A_1=\{x,y\}$ and let $B=B_1$ be such that every path connecting $x$ and $y$ intersects $B$. Then there are no currents $\n\in \Omega_{A \cup B}$ such that $A \ncon{\omega} B$. Hence by the identity above we have $\langle  \varphi_{x}  \varphi_y  \tilde \varphi_B \rangle=0$.
A similar idea, but implemented in the context of planar topology, will be used in the next section to show that the correlations of $\varphi$ and $\tilde \varphi$ vanish
if the insertions of spins are properly chosen on the boundary of a planar graph.

\section{Planarity}\label{sec:planar}
In this section we consider a finite connected planar graph $G=(V,E)$ embedded in the plane. 
We will study correlation functions of the form
\begin{align*} \label{eq:corform}
\langle \varphi_{A_1} \varphi^2_{A_2} \tilde \varphi_{B_1} \tilde \varphi^2_{B_2}\rangle 
\end{align*}
as in Proposition~\ref{lem:RB}, where $A_1,A_2,B_1,B_2$ all lie on the outer face of $G$. 
In order to state our results, we need to introduce a fair number of definitions:
The vertices in $A_2 \cup B_2$ will be referred to as \emph{doubled}. We define $\nod$ to be the set of vertices with multiplicities, i.e., the set where each doubled vertex is included as two copies $v^<$ and $v^>$. We will refer to the elements of $\nod$ as \emph{nodes}.
We assume that $|\nod|$ is even (otherwise the correlation function vanishes), and fix a counterclockwise order $v_1,v_2,\ldots, v_{2n}$ on 
$\nod$ which agrees with the placement of the nodes on the boundary, and where for each doubled node $v$, the copy $v^>$ comes immediately after~$v^<$.
We split the nodes $\nod$ into \emph{even} $\nod_e$ and \emph{odd} $\nod_o$ according to the index in this order.

We will write $A=A_1\cup A_2$ and $B=B_1 \cup B_2$, and we will think of each node in $\nod$ that comes from $A$ (resp.\ $B$) as \emph{red} (resp.\ \emph{blue}). We write $\red$ and $\blue$ for the set of red and blue nodes. We will say that the choice of $\blue$ (and hence automatically $\red=\nod\setminus \blue$) 
is a \emph{coloring} of the nodes.
Finally we partition $\nod$ into \emph{sources} $\source_+$ and \emph{sinks} $\source_-$ by the rule
\[
\origin= (\nod_o\cap \blue) \cup (\nod_e\cap \red) \quad \text{and} \quad  \sink= (\nod_e\cap \blue) \cup (\nod_o\cap \red).
\]
This definition implies that as one goes along the boundary, the nodes alternate between sources and sinks as long as their color does not change, and they keep the same orientation (sink or source)
whenever the color changes.
We say that a coloring $\blue$ is \emph{balanced} if the numbers of resulting sources and sinks are equal.

We will consider partitions $\pi$ of $\nod$. We say that two disjoint sets $P,P'\subset \nod$ cross if one can find $u,v\in P$ and $u',v'\in P$
such that $u<u'<v<v'$ or $u'<u<v'<v$ in the order defined above.
An element $P\in\pi$ of a partition will be called a component. A partition $\pi$ is \emph{planar} (or \emph{noncrossing}) if no two components of $\pi$ cross.
A partition is called \emph{even} if each of its components contains an even number of nodes (possibly zero).  A pairing is a partition whose components all have two elements.
We say that that $\pi$ is \emph{compatible with} a coloring $\blue$ if all nodes in every component of $\pi$ are of the same color.
A coloring $\blue$ is \emph{realizable} if there exists an even planar partition that is compatible with $\blue$.

We are now able to state the first preliminary result.
\begin{figure}
		\begin{center}
			\includegraphics[scale=1.1]{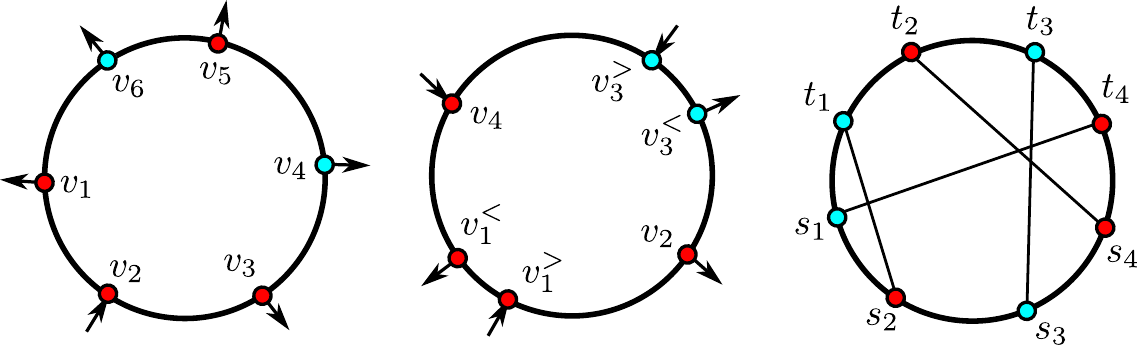}  
		\end{center}
		\caption{Left: an unbalanced coloring corresponding to $\langle \varphi_{v_1}\varphi_{v_2} \varphi_{v_3}\tilde \varphi_{v_4} \varphi_{v_5} \tilde \varphi_{v_6}  \rangle=0$. The arrows pointing inside (resp.\ outside) represent sources (resp.\ sinks). In particular $\origin=\{v_2\}$. 
		Middle: a balanced coloring corresponding to 
		$\langle \varphi^2_{v_1}\varphi_{v_2} \tilde\varphi^2_{v_3} \varphi_{v_4} \rangle$. Here $A\cup B=\{v_1,v_2,v_3,v_4\}$ and $\nod=\{v^<_1,v^>_1,v_2,v^<_3,v^>_3,v_4\}$. Right: the balanced coloring from Corollary~\ref{cor:tot},
		and a pairing $\pi\in \Pi(\origin,\sink)$ with $\textnormal{xg}(\pi)=4$,
		where $\origin =\{s_1,s_2,s_3,s_4\}$ and $\sink =\{t_1,t_2,t_3,t_4\}$}
	\label{fig:color}
\end{figure}

\begin{lemma}\label{lem:bal}
A coloring $\blue$ is realizable if and only if it is balanced. 
\end{lemma}
\begin{proof}
Fix a realizable coloring $\blue$ and an even planar partition $\pi$ of $\nod$ that is compatible with $\blue$. Note that for each component $K$ of $\pi$,
the nodes alternate between sources and sinks as one goes around the outer face. Indeed, all these nodes have the same color since $\pi$ is compatible with $\blue$. Moreover
their parity must alternate since, by planarity, all the nodes between two consecutive nodes of $K$ must be matched together by $\pi$, and all components of $\pi$ contain an even number of sources.
This means that each component of $\pi$ has the same number of sources and sinks, and hence $\blue$ is balanced.

Now fix a balanced coloring $\blue$. We will construct a planar pairing that is compatible with $\blue$. 
We can always find a pair of consecutive nodes $v,v'$ such that $v\in \origin$ and $v'\in \sink$. By definition, they must be of the same color. 
We can hence pair them up, and iterate this procedure for the set of nodes $\nod'=\nod\setminus \{v,v'\}$ which again contains the same number of sources and sinks.
At the end, the resulting pairing is planar by construction.
\end{proof}

Note that each current $\n=(\omega,\eta)\in \Omega_{A_1\cup B_1}$ induces an even planar partition $\pi(\n)$. 
Moreover if $A \ncon{\omega} B$, then $\pi(\n)$ is also compatible with the corresponding coloring~$\blue$.
Recall that $\pev(\nod)$ is the set of subsets of $\nod$ of even cardinality.
\begin{corollary}[Unbalanced colorings]\label{cor:pf}
Let $\blue$ be an unbalanced coloring of~$\nod$.
Then for all coupling constants $J$ and $U$,
\begin{align} \label{eq:omg}
\langle \varphi_{A_1} \varphi^2_{A_2} \tilde \varphi_{B_1} \tilde \varphi^2_{B_2}\rangle =\frac{1}{2^{|\nod|}}\sum_{S \in\pev( \nod)} (-1)^{|S\cap \blue|} \langle\sigma_{\nod \setminus S} \tilde \sigma_S \rangle= 0.
\end{align}
Here, when evaluating the correlation function we naturally project the nodes from $\nod$ onto the corresponding vertices, i.e., $\sigma_{v^<} = \sigma_{v^>}=\sigma_v$ for every doubled vertex~$v$.
If $U=0$, then $\sigma$ and $\tilde \sigma$ are independent and identically distributed, and in particular 
\[
\langle \sigma_{\nod\setminus S} \tilde \sigma_S\rangle=\langle \sigma_{\nod\setminus S}\rangle\langle \sigma_S\rangle.
\]
Moreover, if $|\blue|$ is even, then the second equality in~\eqref{eq:omg} can be rewritten as
\begin{align}\label{eq:pf}
2\langle \sigma_{\nod}\rangle = \sum_{S \in\pev( \nod)\setminus \{\emptyset,\nod\}} (-1)^{|S\cap \blue|+1} \langle \sigma_S\rangle\langle \sigma_{\nod \setminus S}\rangle.
\end{align}
\end{corollary}
\begin{proof}
By Lemma~\ref{lem:bal}, the right-hand side of the identity from Proposition~\ref{lem:RB} is zero since there exist no currents $\n\in \Omega_{\nod}$ such that $\pi(\n)$ is compatible with $\blue$.
\end{proof}

This result says that planar topology causes certain correlations to vanish. More precisely we established that out of the possible $4^{|A\cup B|}$ correlations of the form $\langle \varphi_{A_1} \varphi^2_{A_2} \tilde \varphi_{B_1} \tilde \varphi^2_{B_2}\rangle$, only 
those corresponding to balanced colorings are possibly nonzero.

Note that \eqref{eq:pf} is a recurrence relation since the terms on the right-hand side satisfy $|S|<|\nod|$. 
Therefore it can be used to express many-point correlation functions in terms of the respective two-point functions.
In the next section we will show that the same identities are satisfied by Pfaffians.

\begin{figure}
		\begin{center}
			\includegraphics[scale=1.1]{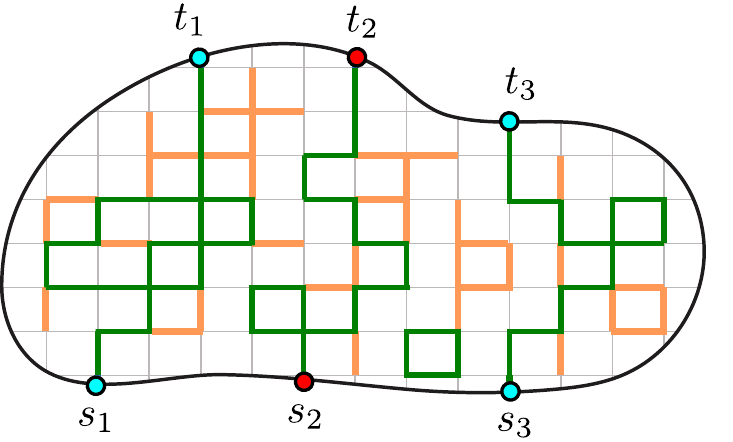}  
		\end{center}
		\caption{
		A current $\n=(\omega,\eta)$ with $\omega\in \mathfrak{P}_{\origin, \sink}$ where $\origin=\{s_1,s_2,s_3\}$ and 
		$\sink=\{t_1,t_2,t_3 \}$. The green edges represent $\eta$ and the orange edges represent $\omega\setminus \eta$ }
	\label{fig:parallel}
\end{figure}

\begin{corollary}[A balanced parallel coloring]\label{cor:tot}
Let $\origin=\{ s_1,s_2,\ldots,s_n \} $ and $\sink=\{ t_1, t_2, \ldots t_n\} $ be a partition of $\nod$ such that 
\begin{align*}
s_1,s_2,\ldots,s_n, t_n, t_{n-1}, \ldots t_1 
\end{align*}
is a counterclockwise order on $\nod$ around the outer face.
The corresponding coloring $\blue$ is shown in Fig.~\ref{fig:color} and Fig~\ref{fig:parallel}.
We define
\[
\frak P_{\origin, \sink} = \{\omega\in \mathfrak F_{\mathcal N}:  s_i \con{\omega} t_i \text{ for all } i , \  s_i \ncon{\omega} t_j \text{ for all } i\neq j\},
\]
where $ s_i \con{\omega} t_j$ (resp.\ $ s_i \ncon{\omega} t_j$) means that $s_i$ and $t_i$ are (resp.\ not) connected by a path of edges in $\omega$.
Then for topological reasons, we have that $\frak P_{\origin, \sink}=\{\omega \in \mathfrak F_{\mathcal{N}}: \red \ncon{\omega}\blue\}$,
and therefore
\begin{align*} 
\langle \varphi_{A_1} \varphi^2_{A_2} \tilde \varphi_{B_1} \tilde \varphi^2_{B_2}\rangle&=\frac1{2^{2n}}\sum_{S \in \pev (\nod)} (-1)^{|S\cap \blue|} \langle  \sigma_{\nod \setminus S}\tilde \sigma_S\rangle \\
&= \frac{1}{Z_{\emptyset}2^{n}}\sum_{\n=(\omega,\eta)\in \Omega_{\nod}}w_{\at}(\n)\mathbf 1\{ \omega \in \mathfrak P_{\origin,\sink}\}.
\end{align*}
\end{corollary}
In particular, if the weights of currents are nonnegative, then so is the above signed sum of correlations of $\sigma$ and $\tilde \sigma$. We will later prove that in the noninteracting case $U= 0$, these inequalities are equivalent to total positivity of certain matrices of boundary two-point functions.

\section{Pfaffians and determinants} \label{sec:pfdet}
In this section we will show that for two independent Ising models, the multi-point correlations of $\sigma$ and $\tilde \sigma$ are Pfaffians and those of $\varphi$ and $\tilde \varphi$ are determinants of the two-point functions.

To this end, we will consider matrices indexed by the nodes $\nod$ with the order defined in the previous section.
Let~$K$ be the $\nod\times\nod$ square antisymmetric matrix given 
by 
\[
K_{u,v} =  \mathbf 1\{u\neq v \}(-1)^{\mathbf 1\{v> u \}} \langle \sigma_u\sigma_v\rangle, \qquad u,v\in\nod.
\]
Again, when evaluating the correlations we project the nodes $\nod$ onto the underlying vertices. In particular, if a vertex $v$ is doubled, then
$\sigma_{v}=\sigma_{v^<}=\sigma_{v^>}$.
For every $S\in \pev(\nod)$, we define ${K}^{S}$ to be the restriction of $K$ to the rows and columns indexed by $S$.
Recall that the Pfaffian of ${K}^{S}$ is the square root of its determinant. It is well known that it can be 
written as
\begin{align} \label{eq:pfcr}
\pf ({K}^{S}) = \sum_{\pi \in \Pi( 
S)} (-1)^{\textnormal{xg} (\pi)}  \prod_{uv \in \pi}\langle \sigma_u\sigma_v \rangle,
\end{align}
where $\Pi( S)$ is the set of all pairings of $S$, i.e., partitions of $S$ into sets of size two.
To define the sign $(-1)^{\textnormal{xg} (\pi)}$, one can think of a diagrammatic 
representation of~$\pi$, where the points representing $\nod$ are placed in the counterclockwise order on the boundary of a disk and straight line segments connect the points inside the disk according to~$\pi$ (see Fig.~\ref{fig:color}). Then $\textnormal{xg}(\pi)$ is the number of pairs in $\pi$ for which the corresponding segments cross. 

We now turn our attention to determinants. Let $\origin \cup \sink$ be a partition of $\nod$ into sources and sinks with $|\origin|=|\sink|$,
and let $\Pi( \origin, \sink)\subseteq\Pi(\nod)$ be the set of pairings in which each pair contains one source and one sink. Note that $\Pi(\origin, \sink )$ can be identified with the set of bijections from $\origin$ to $\sink$. 
In analogy with Postnikov's boundary measurement matrices~\cite{postnikov}, we define the $\origin\times \sink$ square matrix $K^{\origin,\sink}$ by
\begin{align}\label{eq:defK}
K^{\origin,\sink}_{u,v}= (-1)^{s(u,v)} \langle \sigma_u\sigma_v\rangle  \qquad{u\in \origin, v\in \sink},
\end{align}
where $s(u,v)$ is the number of sources strictly between (the smaller and the larger vertex) $u$ and $v$ in the fixed order on $\nod$.
 \begin{example}
Consider the correlator $\langle \varphi^2_{v_1}\varphi_{v_2} \tilde\varphi^2_{v_3} \varphi_{v_4} \rangle$ as in Fig.~\ref{fig:color}. We have $\origin=\{v_1^>,v_3^>,v_4\}$, $\sink= \{v_1^<,v_2,v_3^<\}$ and
\begin{align*}
K^{\origin,\sink}= \begin{pmatrix}
 1 & \langle \sigma_{v_1} \sigma_{v_2} \rangle & \langle \sigma_{v_1} \sigma_{v_3} \rangle \\
- \langle \sigma_{v_3} \sigma_{v_1}\rangle & \langle \sigma_{v_3} \sigma_{v_2} \rangle & 1 \\
 \langle \sigma_{v_4} \sigma_{v_1} \rangle &- \langle \sigma_{v_4} \sigma_{v_2} \rangle &-\langle \sigma_{v_4} \sigma_{v_3} \rangle\\
\end{pmatrix}.
\end{align*}
The $1$'s appearing in the matrix come from the fact that $\langle\sigma_v^{<}\sigma_v^{>}\rangle=\langle \sigma_v^2\rangle=1$.
\end{example}

An analogous formula to \eqref{eq:pfcr} is valid also for determinants as was shown in~\cite{postnikov}. Namely, we have
 \begin{align}\label{eq:det}
 \det (K^{\origin,\sink}) = \sum_{\pi \in \Pi(\origin,\sink)} (-1)^{\textnormal{xg}(\pi)}  \prod_{uv \in \pi}\langle \sigma_u\sigma_v\rangle.
 \end{align}

In order to state the next result we first need to account for some signs.
This is a well known observation and we give a proof here for the sake of completeness (for an alternative proof see e.g.~\cite[Lemma 8.10]{GaPy}).
\begin{lemma}\label{lem:signs}
Let $S\in \pev(\nod)$, $\pi\in \Pi(S)$ and $\pi'\in \Pi(\nod \setminus S)$. Then
\[
 (-1)^{\textnormal{xg}(\pi\cup\pi')} =(-1)^{\textnormal{xg}(\pi)+\textnormal{xg}(\pi') -|S|/2-|S\cap \nod_e|}.
\]
\end{lemma}
\begin{proof}
We first note that $(-1)^{\textnormal{xg}(\pi\cup\pi')-\textnormal{xg}(\pi)-\textnormal{xg}(\pi')} = (-1)^{\textnormal{xg}(\pi,\pi')}$ where we define $\textnormal{xg}(\pi,\pi')$ to be the number of pairs $uv\in \pi$ and $u'v'\in \pi'$ 
that cross. 
We claim that the right-hand side depends only on $S$ and not on $\pi$ and $\pi'$.
This is true since replacing two pairs $uv,u'v'\in \pi$ by $uu',vv'$ or $uv',u'v$ does not change the parity of $\textnormal{xg}(\pi,\pi')$ (this can be checked by considering a small number of cases).
We can hence assume that $\pi$ and $\pi'$ have no crossings and match consecutive vertices in $S$ and $\nod \setminus S$ respectively.
Now we notice that the statement is clearly true if $|S|$ is a set of consecutive nodes since then $|S|/2+|S\cap \nod_e|$ is even, and no pair from $\pi$ crosses a pair from $\pi'$.
It is therefore enough to check that both sides of the equality change in the same way when a node in $S$ is replaced by a neighboring node.
Indeed, this transformation changes the parity of both $\textnormal{xg}(\pi,\pi')$ and $|S\cap \nod_e|$.
\end{proof}

The next identity expresses the determinant of $K^{\origin,\sink}$ as a signed linear combination of products of Pfaffians of $K^S$.
It is likely that this formula is known to experts. However, we were not able to find a suitable reference, and we present its proof as it bears a
strong resemblance to the proof of the second switching lemma in Proposition~\ref{lem:RB}.
\begin{proposition} \label{lem:sousin}
Let $\blue$ be a coloring of $\nod$. Then
\[
\sum_{S\in\pev( \nod)} (-1)^{|S\cap \blue|} \pf(K^{S}) \pf (K^{\nod \setminus S})= \begin{cases}2^{|\nod|/2}\det (K^{\origin,\sink} )& \text{if } |\origin|=|\sink|, \\
0 &  \text{otherwise},
\end{cases}
\]
where $\origin$ and $\sink$ are the sources and sinks associated with $\blue$.
\end{proposition}

\begin{proof} 
By definition of $\origin$ we have
\[
(-1)^{|S\cap \blue|} =(-1)^{ |S\cap \nod_o\cap \blue| + |S\cap\nod_e\cap \red| -|S\cap \nod_e|} = (-1)^{|S\cap \origin|-|S\cap \nod_e|}.
\]
Hence, by \eqref{eq:pfcr} and Lemma~\ref{lem:signs} we can write
\begin{align*}
&\sum_{S\in\pev( \nod)} (-1)^{|S\cap \origin|-|S\cap \nod_e|} \pf ({K}^{S})\pf ({K}^{\nod\setminus  S}) \\
& =  \sum_{S\in\pev( \nod)}(-1)^{|S|/2+|S\cap\origin|}\mathop{\sum_{\pi \in \Pi(S)}}_{\pi'\in \Pi(\nod \setminus S)}(-1)^{\textnormal{xg} (\pi)+\textnormal{xg} (\pi')-|S|/2-|S\cap \nod_e|}\prod_{uv\in \pi\cup \pi'}
\langle \sigma_u\sigma_v\rangle \\
& =  \sum_{S\in\pev( \nod)} (-1)^{|S|/2+|S\cap\origin|}\mathop{\sum_{\pi \in \Pi(\nod)}}_{\pi \textnormal{ comp.\ with } S}(-1)^{\textnormal{xg} (\pi)}\prod_{uv\in \pi}\langle \sigma_u\sigma_v\rangle\\
&= {\sum_{\pi \in \Pi(\nod)}}   (-1)^{\textnormal{xg} (\pi)}\prod_{uv\in \pi}\langle \sigma_u\sigma_v\rangle \Big(\mathop{\sum_{S\in\pev( \nod)}}_{\pi \textnormal{ comp.\ with } S}  
(-1)^{|S|/2+|S\cap\origin|}\Big),
\end{align*}
where we say that $\pi$ is compatible with $S$ if $\pi$ \emph{does not} match a vertex from~$S$ with a vertex from $\nod \setminus S$. To finish the proof it is enough to use~\eqref{eq:det},
and show that
\[
\mathop{\sum_{S\in\pev( \nod)}}_{S \textnormal{ comp.\ with } \pi}  (-1)^{|S|/2+|S\cap\origin|}= 2^{|\nod|/2}\mathbf{1}\{\pi\in \Pi(\origin, \sink) \}.
\]
Indeed, selecting a compatible set $S$ for $\pi$ is equivalent to deciding for each pair in $\pi$ if it is contained in $S$ or not. The total number of such choices is therefore $2^{|\nod|/2}$. Now if $\pi\in \Pi( \origin, \sink)$, then each pair in $\pi$ contains one source and one sink, and then the sign in the sum is constant and equal one. On the other hand if there is a pair $\{u,v\}\in \pi$ such that either 
$\{u,v\}\subseteq \origin$ or $\{u,v\}\subseteq\sink$, then $S\mapsto S\triangle\{u,v\}$ is a sign reversing involution on sets $S$ of even cardinality that are compatible with~$\pi$.
\end{proof}

\begin{corollary} \label{cor:recPf}
In direct analogy with~\eqref{eq:pf}, if $\blue$ is unbalanced and $|\blue|$ is even, then
\begin{align*} 
2\pf (K^{\nod}) = \sum_{S\in\pev( \nod)\setminus \{ \emptyset, \nod \}} (-1)^{|S\cap \blue|+1}\pf (K^S) \pf(K^{\nod \setminus S}).
\end{align*}
\end{corollary}
We can readily rederive the classical result of Groeneveld, Boel and Kasteleyn.
\begin{corollary}[Pfaffian formula for boundary correlations~\cite{GBK}] \label{thm:pf}
Let $J$ be arbitrary and $U= 0$. Then for all $S\in \pev (\nod)$,
\[
\langle \sigma_{S}\rangle = \pf (K^{S}).
\]
\end{corollary}
\begin{proof}
We argue by induction on the (even) cardinality of $S$. 
The case of $|S|=2$ is obvious.
We can hence assume that the statement holds true for every $S\in \pev(\nod)$ with $|S|<2k$. 
Then for any $T\in \pev(\nod)$ with $|T|=2k$, we can use the recursion relation~\eqref{eq:pf} (with $T$ in place of $\nod$)
where we take an unbalanced coloring $\blue$ of $T$ with $|\blue|$ even (there always exists one). By Lemma~\ref{lem:signs} the relation~\eqref{eq:pf} is the same as the
one for Pfaffians from Corollary~\ref{cor:recPf}, and the right-hand side involves only sets $S$ with $|S|<2k$. This concludes the proof.
\end{proof}

The following determinantal formula for the correlations of $\varphi$ and $\tilde \varphi$ is one of the main new contributions of this article.
\begin{theorem} \label{thm:varphi}
Let $J$ be arbitrary and $U= 0$.Then
\begin{align*}
  \langle \varphi_{A_1} \varphi^2_{A_2} \tilde \varphi_{B_1} \tilde \varphi^2_{B_2} \rangle =\begin{cases}2^{-|\nod|/2}\det (K^{\origin,\sink} )& \text{if } |\origin|=|\sink|, \\
0 &  \text{otherwise},
\end{cases}
\end{align*}
where $K^{\origin,\sink}$ is defined in \eqref{eq:defK}.
In particular the above determinant is nonnegative for coupling constants as in~\eqref{eq:nonnegative}.
\end{theorem}
\begin{proof}
The case $|\origin| \neq |\sink|$ follows from Corollary~\ref{cor:pf}, and hence we assume that $|\origin| = |\sink|$.
We first prove the statement in the case with no doubled nodes, i.e., $A_2\cup B_2=\emptyset$. Then $B=\mathcal B$, and by Corollary~\ref{thm:pf} and Proposition~\ref{lem:sousin} we have
\begin{align*}
\langle \varphi_{A} \tilde\varphi_{B} \rangle&={2^{-|\nod|}} \sum_{S\in \pev(\nod)}(-1)^{|B\cap S|} \langle \sigma_S\rangle \langle \sigma_{\nod \setminus S}\rangle \\
&={2^{-|\nod|}}  \sum_{S\in \pev(\nod)}(-1)^{|\blue \cap S|} \pf (K^S) \pf (K^{\nod \setminus S} ) \\
&= {2^{-|\nod|/2}} \det (K^{\origin,\sink} ).
\end{align*}
Now assume that there are doubled nodes, i.e., $A_2\cup B_2\neq \emptyset$. In this case we append two additional vertices $v^<$ and $v^>$ to each doubled node~$v$, and set the 
coupling constants to $J_{vv^<}=J_{vv^>}=m$. For each doubled node~$v$, we then replace each term $\varphi^2_v$ by $\varphi_{v^<}\varphi_{v^>}$ in the correlators, and we proceed likewise for $\tilde \varphi^2_v$.
We then apply the result already proved for the case $|\origin| = |\sink|$, and take the limit $m\to \infty$ in which $\varphi_v=\varphi_{v^<}=\varphi_{v^>}$ almost surely. 
\end{proof}

We now turn our attention to total positivity of boundary two-point functions that was first described in~\cite{LisT}. 
Recall that a square matrix is \emph{totally positive} (resp.\ \emph{nonnegative}) if all its minors are positive (resp.\ {nonnegative}).
Let $ \origin$, $\sink$ be as in Corollary~\ref{cor:tot}. 
Define the $n\times n$ matrix
\begin{align*} 
M_{i,j} = \langle \sigma_{s_i} \sigma_{t_j} \rangle, \qquad 1\leq i,j\leq n.
\end{align*}
For $I,J\subseteq \{1,\ldots,n\}$ we denote by $M^{I,J}$ the restriction of $M$ to rows indexed by $I$ and columns indexed by $J$.

\begin{corollary}[Total positivity of boundary correlations~\cite{LisT}]
For arbitrary $J$ and $U= 0$, the matrix $M$ as defined above is totally nonnegative. Moreover, 
if $I,I'\subseteq \{1,\ldots,n\}$ with $|I|=|I'|=k$, then $\det M^{I,I'}>0$
if and only if there exist $k$ vertex-disjoint paths in $G$ that connect in pairs the sources from $\{s_i \}_{i\in I}$ with the sinks in $\{ t_j\}_{j\in I'}$.
\end{corollary}
\begin{proof}
Let $S_I=\{s_i \}_{i\in I}$ and $T_{I'}=\{ t_j\}_{j\in I'}$. One can check that $\det M^{I,I'}=\det K^{S_I,T_{I'}}$ as the signs in the definition of $ K^{S_I,T_{I'}}$ have the same impact on the determinant as the reversed order of columns in $M^{I,I'}$.
To finish the proof it is therefore enough to combine Corollary~\ref{cor:tot} and Theorem~\ref{thm:varphi}, and the fact that a collection of disjoint paths as in the statement of the corollary defines a current $\n=(\omega,\eta)\in \mathfrak P_{S_I, T_{I'}}$, where $\omega =\eta$ is the union of these paths.
\end{proof}

\begin{remark} \label{rem:fermion}
\emph{Dirac} fermions, unlike \emph{Majorana} fermions, are particles which are different from their antiparticle.
One can imagine that the vertices on the outer boundary of $G$ represent fermions.
To each such fermion $v$, there correspond two operators of interest:  the creation operator $a_v$ and the annihilation operator $a_v^{\dagger}$.
For two operators $a_v,b_u$, define the anticommutator by $\{a_u,b_v\}=a_ub_v+b_va_u$.
We then have the fermionic anticommutation relations
\[
\{ a_u,a_v^{\dagger} \} =\delta_{u,v}, \quad \text{and} \quad \{ a_u,a_v\} =\{a_u^{\dagger},a_v^{\dagger}\}=0.
\]
In view of the results above, the correlations of $\varphi$ and $\tilde \varphi$ are a model for expectation values of Dirac fermions which are noninteracting in the case $U=0$.
Indeed the insertion of $\varphi_v$ or $\tilde \varphi_v$ in the correlator corresponds to the insertion of either the creation operator $a_v$ or the annihilation operator $a_v^{\dagger}$ into the expectation value, depending on whether these
insertions yield a source or a sink respectively (this depends on the preceding insertions). Moreover, these values are nonzero if and only if the number of creation and annihilation operators are equal,
and the above anticommutation relations follow from the fact that $\varphi^2_v+\tilde \varphi_v^2=1$ and $\varphi_v\tilde \varphi_v=0$ respectively.
\end{remark}

\begin{remark}\label{rem:wick}
The picture presented here for two i.i.d.\ Ising models should be compared to the one of two i.i.d.\ real Gaussian fields $\phi=(\phi_v)_{v\in V}$ and $\tilde \phi=(\tilde \phi_v)_{v\in V}$,
and the combined complex Gaussian field $\Phi=(\phi+i\tilde \phi)/\sqrt2$. In this case the moments of $\phi$ and  $\Phi$ are governed by the \emph{bosonic} Wick rules, and are given by \emph{hafnians} 
and \emph{permanents} of their two-point functions respectively. More precisely for two (multi-) sets $A,B$ of vertices of the same cardinality, we have 
\[
\mathbb E\Big[\prod_{v\in A}\phi_v\Big] = \sum_{\pi\in \Pi(A)} \prod_{uv\in \pi} \mathbb E[\phi_u\phi_v],
\ \mathbb E\Big[\prod_{v\in A}\Phi_v\prod_{v\in B}\overline{\Phi}_v\Big] = \hspace{-0.2cm}\sum_{\pi\in \Pi(A,B)} \prod_{uv\in \pi} \mathbb E[\phi_u\phi_v],
\]
where $\mathbb E$ is the expectation with respect to the Gaussian measure, and $\overline \Phi$ is the complex conjugate of $\Phi$.
Unlike for Ising models, these identities are valid for any finite graph $G$ and any choice of vertices $A$ and $B$.
\end{remark}

\begin{remark}
Aizenman et al.~\cite{ADTW} derived an asymptotic version of the Pfaffian formula for the critical Ising model on graphs embedded in the upper half-plane where planarity is broken by allowing edge crossings (for the exact statement we refer the reader to~\cite{ADTW}). These edge crossings 
make the arguments above not valid on the level of exact vanishing of correlations corresponding to unbalanced colorings. However, the fact that the
critical currents have a fractal structure on large scales, causes many connectivities, that are deterministically forced to happen in the planar case,
to happen with high probability for graphs with edge crossings. Since the graphical representations of critical Ashkin--Teller model are also expected to be fractal~\cite{IkRa,IkRa1},
a similar phenomenon as~\cite{ADTW} should hold true at criticality for $U\neq 0$. Without going into technical details, we expect that the ratio of two 
critical correlation functions $\langle \varphi_{A_1} \varphi^2_{A_2} \tilde \varphi_{B_1} \tilde \varphi^2_{B_2} \rangle/   \langle \varphi_{A_1} \varphi^2_{A_2}  \varphi_{B_1}  \varphi^2_{B_2}  \rangle$
should be close to zero whenever the coloring corresponding to the numerator is unbalanced, and the points $A\cup B$ on the boundary are pairwise far away from each other.
\end{remark}

\section{Order-disorder correlations} \label{sec:disorders}
In this section we still consider the planar setup and we would like to derive analogous linear relations for correlators evaluated not only on the boundary but also in the bulk.
If we restrict ourselves to spin correlations, then clearly the arguments from before are not valid as it is not anymore possible to force different clusters of the current to intersect (they can evade 
each other by going {around} the insertion points). The well known remedy in the noninteracting case is to consider \emph{disorder operators} of Kadanoff and Ceva~\cite{KC} which effectively change the topology of $G$ to 
that of a branched double cover.
These operators, as the spins (which are also called \emph{order operators}), will come in two types $\mu$ and $\tilde \mu$ and will be evaluated not at the vertices but rather at the faces of~$G$.

Let $C, D$ be two sets of faces and let $\gamma_u$, $u\in C\cup D$, be fixed simple dual
paths connecting the faces to the outer face (we could as well choose any other face where these paths jointly end).
Let $\Gamma_C\subseteq E$ be the set of edges dual to $\gamma_{u_1}\triangle\cdots\triangle \gamma_{u_k}$ where $C=\{u_1,\ldots,u_k \}$.
Similarly define $\Gamma_D$, and let $\epsilon_{e}=(-1)^{\mathbf 1\{ e\in \Gamma_C\}}$ and $\tilde \epsilon_{e}=(-1)^{\mathbf 1 \{e\in \Gamma_D \}}$.
We consider the Ashkin--Teller probability measure with \emph{disorder insertions} at $C$ and $D$, given by
\begin{align}\label{def:ATdis}
\IP_{C,D} (\sigma,\tilde \sigma)=\frac1{\mathcal Z^{C,D}}\prod_{e\in E} \exp\big(J_{e} (\epsilon_{e}\sigma_e+\tilde \epsilon_{e}\tilde \sigma_e)+U_{e} (\epsilon_{e}\tilde\epsilon_{e}\sigma_e\tilde \sigma_e+1)\big),
\end{align}
where we write $\sigma_{uv}=\sigma_u\sigma_v$.
We note that this measure depends not only on the chosen faces $C,D$ but also implicitly on the underlying paths.
We will study the joint order-disorder correlators defined by 
\[
\langle  \sigma_A \tilde \sigma_{B}\mu_C\tilde \mu_D\rangle = \langle  \sigma_A \tilde \sigma_{B}\rangle_{C,D} \frac{\mathcal Z^{C,D}}{\mathcal Z},
\] 
where $\langle \cdot \rangle_{C,D}$ is the expectation with respect to $\IP_{C,D}$, and $A,B$ are sets of vertices as in previous sections.
We note that by planar duality~\cite{Fan} (see also~\cite{PfiVel}), the disorder operators become order operators for the dual Ashkin--Teller model.

We want to give a random current representation of such mixed correlations. 
A naturally associated notion is that of a \emph{double cover of} $G$ \emph{branching around} $C\triangle D$~\cite{CheSmi,CI,CHI}. 
In general, a double cover of a graph $G$ is a graph $G'$ with a two-to-one local graph isomorphism from $ G'$ to $G$ mapping $v$ to $\underline v$. 
If $v$ is a vertex of $G'$, we denote by $v^\dagger$ the vertex satisfying $v\neq v^\dagger$ and $\underline{v}=\underline{v^\dagger}$, and we say that $v$ and $v^\dagger$ belong to different \emph{sheets} of $G'$.
We also say that $G'$ \emph{branches around} a face $u$ if the cycle composed of the edges surrounding $u$
lifts to a path connecting two vertices in different sheets of~$G'$. Otherwise, if such cycle lifts to a cycle in $G'$, then $G'$ does not branch around~$u$.
If a graph has $m$ faces, then there are $2^m$ double covers corresponding to the sets of faces around which the cover branches.
We denote by $G^{C\triangle D}= (V^{C\triangle D},E^{C\triangle D})$ the double cover branching around $C\triangle D$.
The main reason why we consider such double covers is that the spin configuration $\tau\in \{-1,+1\}^{V^{C\triangle D}}$ satisfying
\begin{align} \label{def:tau}
\tau_u\tau_v=\epsilon_{{uv}}\tilde \epsilon_{uv}\sigma_u \sigma_v \tilde \sigma_u\tilde \sigma_v
\end{align}
is well defined only on $G^{C\triangle D}$ but not on $G$ itself. 
Note that from this definition it follows that
\begin{align} \label{eq:negsheet}
\tau_{v}=-\tau_{v^\dagger}
\end{align}
for every $v\in V^{C\triangle D}$. In other words, $\tau$ has a multiplicative monodromy of $-1$ around every branch point.

To describe the influence of disorder correlators on the topology of currents themselves, we define $\frak F^*_{C\triangle D}$ to be the collection of sets $\omega \subseteq E$ 
such that every cycle contained in $\omega$ surrounds, i.e., disconnects from infinity, an even number of faces in $C\triangle D$. In particular $\frak F^*_{\emptyset}$ is the set of all subsets of~$E$. Also, for a current $\n=(\omega,\eta) \in\Omega_{A \triangle B}$ for which $\omega\in \frak F_B \cap \frak F^*_{C\triangle D}$, we define its \emph{sign with respect to} $A,B$ and $C,D$ by
\begin{align} \label{eq:sign}
\textnormal{sgn}(\n)=\textnormal{sgn}(A,B,C,D;\n)=(-1)^{|\eta\cap \Gamma_C|} (-1)^{|\rho_B \cap (\Gamma_C\triangle \Gamma_D)|},
\end{align}
where $\rho_B \subseteq \omega$ is \emph{any} collection of simple edge-disjoint paths that connect the vertices in~$B$ into pairs. If $B=\emptyset$, then we take $\rho_B=\emptyset$, and otherwise such paths exist since $\omega\in  \frak F_B $. Moreover, this sign is well defined. 
Indeed, if we assume that there exists another collection of paths
$\rho'_B$ yielding a different sign, then immediately
\[
(-1)^{|\rho_B \cap (\Gamma_C\triangle \Gamma_D)|+|\rho'_B \cap (\Gamma_C\triangle \Gamma_D)|}=(-1)^{|(\rho_B\triangle \rho'_B) \cap (\Gamma_C\triangle \Gamma_D)|}=-1,
\] 
which in turn means that $\rho\triangle \rho'$ contains a cycle that surrounds an odd number of faces from $C\triangle D$. 
Indeed, $\rho_B\triangle \rho'_B$ can be written as a union of disjoint cycles, and the equality above implies that one of these cycles must contain an odd number of edges from $\Gamma_C\triangle \Gamma_D$.
By properties of planar topology, we deduce that this cycle surrounds an odd number of faces from $C\triangle D$. This is a contradiction with the fact that $\omega\in \frak F^*_{C\triangle D}$.

We are now able to prove a planar generalization of Proposition~\ref{prop:switch} (which we recover in the absence of disorders, i.e., when $C=D=\emptyset$).
\begin{proposition} \label{prop:switchdis} For all coupling constants $J$ and $U$,
\begin{align*}
\langle  \sigma_A \tilde \sigma_{B}\mu_C\tilde \mu_D\rangle_{\at} = \frac1{Z_{\emptyset}} \sum_{\n=(\omega,\eta)\in\Omega_{A\triangle B} } \textnormal{sgn}(\n) w_{\at}(\n) \mathbf 1\{\omega\in \frak F_B \cap \frak F^*_{C\triangle D}\}.
\end{align*}
\end{proposition} 

 \begin{proof}
Consider the double cover $G^{C \triangle D}$, and let $\tau$ be spins on $V^{C\triangle D}$ satisfying condition \eqref{def:tau}.
As in Proposition~\ref{prop:switch}, for every edge $e=uv\in E^{C\triangle D}$, we have
\begin{align} \label{eq:newweight}
 \exp\big(\tfrac 12J_{\underline{e}} (\epsilon_{{e}}\sigma_{{e}}+\tilde \epsilon_{{e}}\tilde \sigma_{{e}})+\tfrac 1 2 U_{\underline{e}} (\epsilon_{{e}}\tilde\epsilon_{{e}}\sigma_{{e}}\tilde \sigma_{{e}} +1)\big)
=1+ \delta_{\tau_u\tau_v}(a_{{e}}\epsilon_{{e}} \sigma_{{e}}+b_{{e}}  ),
\end{align}
where we write $\sigma_{e}=\sigma_{\underline u} \sigma_{\underline v}$, $\epsilon_e=\epsilon_{\underline e}$, and 
\[
a_{ e}=x_{\underline e}(\tfrac 12U_{ \underline e}, \tfrac 12J_{\underline e}),\ \textnormal{and} \ b_{ e}=y_{ \underline e}(\tfrac 12U_{\underline e}, \tfrac 12J_{\underline e}),
\] 
and where the weights $x$ and $y$ are defined in~\eqref{eq:CD}.

Let $\Sigma=\{-1,+1\}^{V}$ and let $\mathcal T$ be the set of spin configurations $\tau\in\{-1,+1\}^{V^{C\triangle D}}$ satisfying the monodromy condition~\eqref{eq:negsheet}.
Note that for a collection of paths $\rho\subseteq E$ connecting $B$ into pairs as in~\eqref{eq:sign}, we have by definition of $\tau$ that 
\[
\tilde \sigma_B = \sigma_B\tau_{B'} (-1)^{|\rho\cap \Gamma_{C\triangle D}|},
\]
where $B'\subseteq V^{C\triangle D}$ are the endpoints of lifts of the paths in $\rho$. 
We can now rewrite the product of the Gibbs--Boltzmann factors in \eqref{def:ATdis} over the edges in $E$ as a product of their square roots \eqref{eq:newweight} over twice as many edges in $E^{C\triangle D}$.
We therefore have
\begin{align*}
&\mathcal{Z} \langle \sigma_A \tilde \sigma_{B} \mu_C \tilde \mu_D \rangle_{\at} 
= \sum_{\sigma,\tilde \sigma\in \Sigma} \sigma_{ A } \tilde \sigma_B\prod_{e=uv\in  E^{C\triangle D}}\big(1+ \delta_{\tau_u\tau_v}\big(a_{{e}} \epsilon_{{e}}\sigma_{{e}}+b_{{e}}  ) \big)\\
&= \sum_{\sigma\in \Sigma, \tau\in \mathcal T} \sigma_{ A\triangle B}  \sum_{\omega\subseteq  E^{C\triangle D}} \tau_{B'} (-1)^{|\rho\cap\Gamma_{C \triangle D}|}\prod_{e=uv\in \omega} \delta_{\tau_u\tau_v} \delta_{\tau_{u^\dagger}\tau_{v^\dagger}}(a_{{e}}  \epsilon_{{e}}\sigma_{{e}}+b_{{e}}  )\\
&= \sum_{\sigma \in \Sigma}  \sum_{\omega\subseteq E^{C\triangle D}}\mathbf 1\{\underline \omega\in\mathfrak F\}2^{k(\underline \omega)} (-1)^{|\rho\cap \Gamma_{C\triangle D}|}\sigma_{ A\triangle B} \prod_{e=uv\in \omega} ( a_{ e}  \epsilon_{{e}} \sigma_{ e}+b_{ e}  ),
\end{align*}
where $\rho\subseteq \underline \omega$ are paths connecting $B$ into pairs, and where $\mathfrak F=\frak F_B \cap \frak F^*_{C\triangle D}$.
In the second line we used the fact that $ \delta_{\tau_u\tau_v}= \delta_{\tau_{u^\dagger}\tau_{v^\dagger}}\in \{0,1\}$. The indicator $\mathbf1\{\underline \omega\in  \frak F^*_{C\triangle D} \}$ 
in the last line is a consequence of 
property~\eqref{eq:negsheet}. 
Indeed if $\underline \omega$ contains a cycle surrounding an odd number of faces in $C\triangle D$, then such cycle lifts to a path in $G^{C\triangle D}$ connecting vertices on different sheets. 
Hence, by~\eqref{eq:negsheet} the product of $\delta_{\tau_u\tau_v}$ along this path is zero. The factor $2^{k(\omega)}\mathbf 1\{\underline \omega\in \frak F_B\}$ arises for the same reason as in the
case $C=D=\emptyset$.
We now expand the product in the last line and write
\begin{align*}
&\sum_{\sigma\in \Sigma}   \sum_{\omega\subseteq  E^{C\triangle D}}\mathbf 1\{\underline \omega\in\mathfrak F\}2^{k(\underline\omega)} (-1)^{|\rho\cap \Gamma_{C \triangle D}|}\sum_{\eta \subseteq \omega} \sigma_{ A\triangle B\triangle \delta(\eta')} 
\prod_{e\in  \eta} \epsilon_e a_{  e} \prod_{e\in \omega \setminus \eta}b_{ e}  \\
&={2^{|V|}}\mathop{\sum_{\omega\subseteq E^{C\triangle D},\ \eta \subseteq \omega}}_{\delta(\eta')=A\triangle B} \Big( (-1)^{|\rho\cap \Gamma_{C\triangle D}|}  \prod_{e\in \eta} \epsilon_{ e}\Big)
\mathbf 1\{\underline \omega\in\mathfrak F\}
2^{k(\omega)} \prod_{e\in  \eta} a_{ e} \prod_{e\in \omega \setminus \eta}b_{ e}  \\
&= {2^{|V|}}\sum_{\n\in \Omega^{C\triangle D}_{A \triangle B}} \textnormal{sgn}(\n)w_{\at}(\n) \mathbf 1 \{\underline \omega\in \mathfrak F\},
\end{align*}
where $\eta'=\triangle_{e\in \eta} \{\underline e\}$ is the set of edges of $G$ whose exactly one lift to $G^{C\triangle D}$ belongs to $\eta$.
To justify the last line, we note that for every current $\tilde \n=(\tilde \omega,\tilde \eta)$, we have
\begin{align*}
\mathop{\sum_{\omega\subseteq E^{C\triangle D},\ \eta \subseteq \omega}}_{\underline \omega=\tilde\omega,\ \eta'=\tilde\eta} \prod_{e\in  \eta}  a_{  e} \prod_{e\in \omega \setminus \eta}b_{ e} =\prod_{e\in  \tilde\eta}  x_{  e} \prod_{e\in \tilde \omega \setminus \tilde \eta}y_{ e}.
\end{align*}
This can be seen by considering the local configurations for each edge separately and the fact that the weights satisfy $a_e^2+2b_e+b_e^2=y_e$ and $2a_e(1+b_e)=x_e$.
Using that $\mathcal Z=2^{|V|}Z_{\emptyset}$ we finish the proof.
 \end{proof}
 
 We now proceed to study the special case when the disorder and order insertions are placed next to each other. To be more precise,
a \emph{corner} $c=vf$ of $G$ is a pair of a vertex $v$ and a neighbouring face $f$. The \emph{Kadanoff--Ceva fermions}~\cite{KC} are defined as formal insertions of 
\[
 \psi(c)=  \sigma(v)  \mu(f) \qquad \text{and} \qquad \tilde \psi(c)= \tilde \sigma(v) \tilde \mu(f)
\]
into the correlation functions. In analogy to the $\varphi$ and $\tilde \varphi$ variables, we also define 
\[
 \chi(c)= \frac{ \psi(c)+\tilde\psi(c)}2 \qquad \text{and} \qquad \tilde \chi(c)=\frac{\psi(c)-\tilde\psi(c)}2,
\] 
and their squares, where we formally take $\mu^2=\tilde \mu^2=1$.
In what follows we will consider correlators of the form
\[
\langle \chi_{A_1} \chi^2_{A_2} \tilde \chi_{B_1} \tilde \chi^2_{B_2}\rangle
\]
where $A_1, A_2, B_1, B_2$ are disjoint sets of corners.
We will assume henceforth that the dual paths connecting the faces of the corners to the outer face are mutually avoiding, i.e., do not share edges. 
We will write $A=A_1\cup A_2$, $B=B_1\cup B_2$, and
for a set of corners $T$, we define $V(T)$ and $F(T)$ to be the sets of vertices and faces of $T$ respectively. 

This definition is possibly far from intuitive but natural for the statement of the next two results.
\begin{definition}
We say that a current $\n=(\omega,\eta)\in \Omega_{V(A_1\cup B_1)}$ with $\omega\in  \frak F^*_{F(A_1\cup B_1)}$ is \emph{null} for the correlator $\langle \chi_{A_1} \chi^2_{A_2} \tilde \chi_{B_1} \tilde \chi^2_{B_2}\rangle$ if there exists a set of corners $T\in \pev( A\cup B)$ with $\omega\in \frak F_{V(T)}$, 
such that
\begin{align} \label{eq:involution}
 |T\cap B|+|\eta\cap \Gamma_{F(T)}|+|\rho_{V(T)}\cap \Gamma_{F(A\cup B)} | \quad \textnormal{is odd.}
\end{align}
As before, $\Gamma_{F(T)}$ are the edges crossed by fixed dual paths connecting $F(T)$ to the outer face, 
and $\rho_{V(T)}$ is any collection of simple edge-disjoint paths contained in $\omega$ that connects the vertices in $V(T)$ into pairs.
\end{definition}

\begin{center}
\begin{figure} \label{fig:KC}
\includegraphics[scale=0.9]{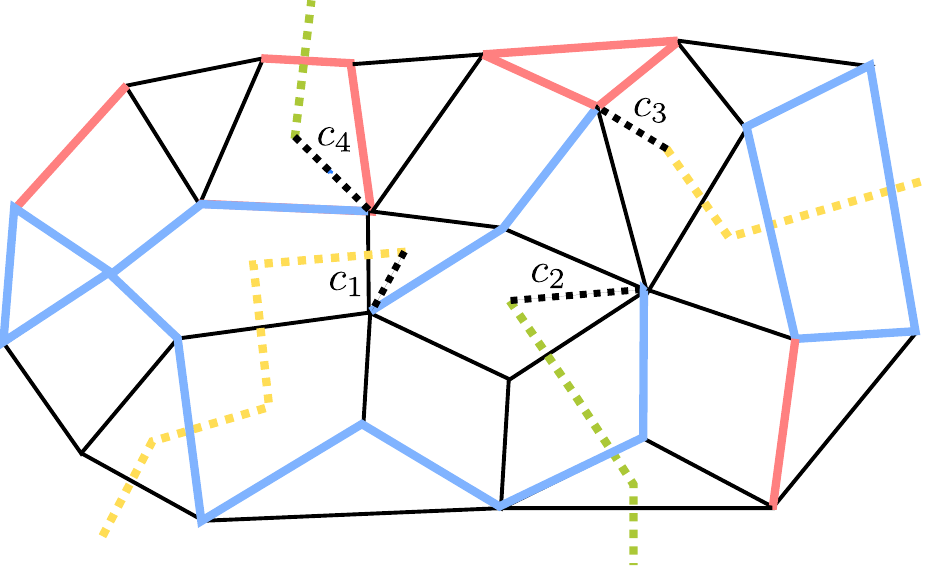}
\caption{
The blue edges represent $\eta$ and the red ones $\omega \setminus \eta$ in a current $\n=(\omega,\eta)$. 
The dual green edges cross $\Gamma_{\{c_2,c_4\}}$ and the yellow ones cross
$\Gamma_{\{c_1,c_3\}}$. 
The dotted black edges represent the corners $c_1,c_2,c_3,c_4$.
The current is null for the correlator $\langle \xi_{c_1}  \tilde\xi_{c_2} \xi_{c_3}\tilde \xi_{c_4}\rangle$ as $T=\{c_1,c_3\}$ satisfies~\eqref{eq:involution}.
In fact, as the coloring $B=\{c_1,c_3\}$ is unbalanced, all currents are null for this correlator by Lemma~\ref{lem:allnull}, and hence $\langle \xi_{c_1}  \tilde\xi_{c_2} \xi_{c_3}\tilde \xi_{c_4}\rangle=0$ by Proposition~\ref{prop:nullexpansion}. Note that $\n$ is not null e.g.\ for $\langle \tilde\xi_{c_1}  \tilde\xi_{c_2} \xi_{c_3} \xi_{c_4}\rangle$}
\end{figure}
\end{center}

\begin{proposition} \label{prop:nullexpansion}
With the above notation, for all coupling constants $J$ and~$U$,
\begin{align*}
\langle \chi_{A_1} \chi^2_{A_2} &\tilde \chi_{B_1} \tilde \chi^2_{B_2}\rangle 
 =\frac{1}{Z_{\emptyset}}\mathop{\sum_{\n=(\omega,\eta)\in \Omega_{V_1} }}_{\n \textnormal{ is not null}}\vspace{-0.4cm}w_{\at}(\n)2^{-k_{V(A\cup B)}(\omega)}(-1)^{|\eta\cap \Gamma_{F_1}|}\mathbf 1\{\omega\in  \frak F^*_{F_1}\},
\end{align*}
where $V_1=V(A_1\cup B_1)$ and $F_1=F(A_1\cup B_1)$.
\end{proposition}
\begin{proof}
Note that $\chi_v^2=\psi_v \chi_v$ and $ \tilde \chi_v^2= \psi_v\tilde \chi_v$.
Therefore, as in the proof of Proposition~\ref{lem:RB}, we obtain
\begin{align*}
\langle \chi_{A_1} \chi^2_{A_2} \tilde \chi_{B_1} \tilde \chi^2_{B_2}\rangle & =2^{-|A\cup B|}\sum_{S \in\pev(A\cup B)} (-1)^{|S\cap B|} \langle  \psi_{(A_1\cup B_1)\triangle S} \tilde\psi_S \rangle.
\end{align*}
Using Proposition~\ref{prop:switchdis},
the right-hand side equals $1/Z_{\emptyset}$ times
\begin{align*}
2^{-|A\cup B|} \hspace{-0.4cm}\mathop{\sum_{\n=(\omega,\eta)\in \Omega_{V_1 } }}_{ \omega \in \frak{ F}^*_{F_1}}\hspace{-0.4cm}w_{\at}(\n)\Big( \sum_{S \in\pev (A\cup B)} (-1)^{|S\cap B|} \textnormal{sgn}(\n)\mathbf 1\{\omega\in \frak F_{V(S)}\}\Big),
\end{align*}
where the dependence of the sign on $S$ is
\[
\textnormal{sgn}(\n)=\textnormal{sgn}(V_1\triangle V(S),V(S),F_1 \triangle F(S),F(S);\n).
\]
It is therefore enough to prove that for every $\n=(\omega,\eta)$ with $\omega \in  \frak{ F}^*_{F_1}$, the second sum is equal to 
\[
2^{|A\cup B|-k_{V(A\cup B)}(\omega)}(-1)^{|\eta\cap \Gamma_{F_1}|}\mathbf 1\{\n\textnormal{ is not null}\} .
\]
To this end, we first assume that $\n$ is null and show that the sum is zero. The existence of $T$ as in \eqref{eq:involution} allows us to construct an involution of the collection of sets $S\in\pev (A\cup B)$ with $\omega\in \mathfrak F_{V(S)}$ that changes the sign $(-1)^{|S\cap B|} \textnormal{sgn}(\n)$. Indeed, consider the map $S\mapsto S\triangle T$.
Then the product of signs corresponding to $S$ and $S\triangle T$ is independent of $S$, and given by 
\[
(-1)^{ |T\cap  B|+|\eta\cap \Gamma_{F(T)}|+|\rho_{V(T)}\cap \Gamma_{F_1} |}=-1,
\] 
which yields the desired cancellation. On the other hand if $\n$ is not null, then $(-1)^{|S\cap B|} \textnormal{sgn}(\n)$ is constant, and for $S=\emptyset$ equal to 
\[
\textnormal{sgn}(V_1,\emptyset,F_1,\emptyset; \n) =(-1)^{|\eta\cap \Gamma_{F_1}|}.
\]
Moreover, $2^{|A\cup B|-k_{V(A\cup B)}(\omega)}$ is the cardinality of the set of all $S\in\pev (A\cup B)$ with $\omega\in \mathfrak F_{V(S)}$.
\end{proof}
\begin{remark}
Consider the situation when $F(A\cup B)$ contains only the outer face, and hence the vertices $V(A\cup B)$ lie on the outer boundary of the graph. Then one can choose the dual paths
so that $\Gamma_{F_1}=\emptyset$. 
In this case a current $(\omega,\eta)$ is null if and only if $A\con{\omega}B$, and the above result recovers Proposition~\ref{lem:RB} in this planar setup.
\end{remark}

To derive linear relations for correlations of Kadanoff--Ceva fermions, we first repeat verbatim the definitions from Section~\ref{sec:planar} for boundary spin correlations.
To be more precise, we call the corners in $A_2\cup B_2$ \emph{doubled}, and 
consider the set of \emph{(corner) nodes} $\mathcal K$, i.e., corners with multiplicities, where each doubled corner $c$ comes in two copies $c^<$ and $c^>$.
As before, we will have to introduce a natural order on $\mathcal K$.
To this end, note that removing the edges in $\Gamma_{F(A\cup B)}$, results in a graph $G'$ for which all the vertices
$V(A\cup B)$ lie on the outer face. We order $V(A\cup B)$ according to a counterclockwise order along this outer face. 
This induces an order on $\mathcal K$ where, for doubled nodes $c$, the copy $c^>$ comes immediately after $c^<$.
We refer to nodes coming from $B_1\cup B_2$ in this construction as blue, and define $\mathcal B$ the set of all blue nodes.
We call $\mathcal B$ a coloring of the nodes. The construction of source and sink nodes is exactly the same as in Section~\ref{sec:planar}.

As for spins on the boundary, we obtain linear relations for the correlations of Kadanoff--Ceva fermions.

\begin{lemma}\label{lem:allnull}
Let $\mathcal K$ be as above and let $\mathcal B \subseteq \mathcal K$ be an unbalanced coloring. Then all currents $\n=(\omega,\eta)\in \Omega_{V(A_1\cup B_1)}$ with $\omega\in \mathfrak F^*_{F(A_1\cup B_1)}$ are null for the corresponding correlator, and as a result 
\begin{align*}
\langle \chi_{A_1} \chi^2_{A_2} \tilde \chi_{B_1} \tilde \chi^2_{B_2}\rangle  =0
\end{align*}
for all coupling constants $J$ and $U$.
\end{lemma}
\begin{proof}
Write $F_1=F(A_1\cup B_1)$ and $V_1=V(A_1\cup B_1)$, and let $\n=(\omega,\eta)\in \Omega_{V_1}$ with $\omega\in \mathfrak F^*_{F_1}$. 
By Proposition~\ref{prop:nullexpansion}, it is enough to prove that $\n$ is null.
To this end, for each vertex $v$ of nonzero even degree in $\eta$, we arbitrarily choose one of the two noncrossing pairings of the edges incident on $v$ in which 
each edge is paired with its neighbour, i.e., an edge incident on the same face. For each vertex of odd degree, we proceed analogously after choosing an edge which will remain unpaired. The reflexive and transitive closure of the relation of being paired together defines a partition of $\eta$. This partition naturally splits into a collection of mutually noncrossing
cycles~$\mathcal C$ and simple
paths $\mathcal P$ with endpoints in $V_1$. We extend the \emph{primal} paths in $\mathcal P$ by attaching to them the corresponding 
\emph{dual} paths from the definition of $\Gamma_{F_1}$ (for each $\{v,f\}\in \mathcal K$, we connect the primal 
path with an endpoint at $v$ with the dual path starting at $f$), and we call $\mathcal P'$ the resulting collection of extended paths starting end ending at the outer face of~$G$.

Note that since the coloring of $\mathcal B \subseteq \mathcal K$ is unbalanced, the numbers of resulting sources and sinks are different, and hence there must exist an extended path $\mathfrak p \in \mathcal P'$ which connects two sources or two sinks with each other.
Let $c=\{v,f\}$ and $c'=\{v',f'\}$ be the corresponding corners that the path $\mathfrak p$ joins. We will prove that $T=\{c,c'\}$ satisfies \eqref{eq:involution} (and hence $\n$ is null) by showing that
\begin{align*}
 |\{c,c'\}\cap B| \quad \text{and} \quad |\eta\cap \Gamma_{\{f,f' \}}|+|\rho_{\{v,v'\}}\cap \Gamma_{F_1} | 
\end{align*}
are of different parity, where $\rho_{\{v,v'\}}$ is chosen to be the primal part of $\mathfrak p$.

We first observe that $|\{c,c' \}\cap B|$ has different parity than $\textnormal{x}(\mathfrak p)$ that we define to be the total number of crossings between the extended path $\mathfrak p$ and all 
the remaining extended paths in $ \mathcal P' \setminus \{ \mathfrak p\}$ (excluding possible self-crossings of~$\mathfrak p$). 
Indeed, $c$ and $c'$ are both sources or both sinks. Hence if they are of different color (and $|\{c,c'\}\cap B|$ is odd), then by definition of sources and sinks, there is an 
even number of nodes (in the fixed order on $\mathcal K$) between $c$ and~$c'$. Analogously, if $c$ and $c'$ are of the same color (and $|\{c,c'\}\cap B|$ is even), then there is an odd
number of nodes between $c$ and $c'$.
On the other hand, the extended paths start and end on the outer boundary of $G$, and therefore planar topology implies that
this number has the same parity as $\textnormal{x}(\mathfrak p)$.

To finish the proof it is therefore enough to show that the parity of $|\eta\cap \Gamma_{\{f,f' \}}|+|\rho_{\{v,v'\}}\cap \Gamma_{F_1} |$ 
is the same as that of the number of crossings $\textnormal{x}(\mathfrak p)$.
To this end, recall that each extended path is composed of a primal and two dual paths.
Also note that by construction the only way the extended paths can cross is when a primal part of one path crosses the dual part of another path.
This implies that $\textnormal{x}(\mathfrak p)$ has the same parity as
\begin{align} 
|(\bigcup \mathcal P\setminus\rho_{\{v,v'\}})\cap \Gamma_{\{f,f' \}}|& + | \rho_{\{v,v'\}}\cap \Gamma_{F_1\setminus \{f,f' \}}| \nonumber  \\
&=|(\bigcup \mathcal P)\cap \Gamma_{\{f,f' \}}| + | \rho_{\{v,v'\}}\cap \Gamma_{F_1}| . \label{eq:crossings}
\end{align}
We finally note that the contribution of $\bigcup \mathcal C\subseteq \eta$ to $|\eta\cap \Gamma_{\{f,f' \}}|$ is even since, again by construction, no cycle in $\mathcal C$ intersects the primal part $\rho_{\{v,v' \}}$ of $\mathfrak p$, and each such cycle crosses the full path $\mathfrak p$ an even number of times.  Hence $|\eta\cap \Gamma_{\{f,f' \}}|$ has the same parity as $|(\eta\setminus \bigcup \mathcal C)\cap \Gamma_{\{f,f' \}}|=|(\bigcup \mathcal P)\cap \Gamma_{\{f,f' \}}|$. This together with \eqref{eq:crossings} finishes the proof.

\end{proof}
This result together with arguments identical to those for boundary spin correlations directly yield the known Pfaffian formulas for the correlations of Kadanoff--Ceva fermions. We refer the reader to~\cite[Section 6]{ADTW} for a brief historical account of Pfaffian identities for order-disorder correlations (see also~\cite{CCK}), and we note that
the first proof that used random currents was obtained by Aizenman et al.\ in~\cite{ADTW}.

In what follows, we will consider matrices indexed by the nodes $\mathcal K$ with the prescribed order.
Let~$F$ be the $\mathcal K\times\mathcal K$ square antisymmetric matrix given 
by 
\[
F_{c,c'} =  \mathbf 1\{u\neq v \}(-1)^{\mathbf 1\{c'> c \}} \langle \psi_c\psi_{c'}\rangle, \qquad c,c'\in\mathcal K.
\]
Recall that the definition of these correlations assumes implicitly for each $c=\{v,f\}\in\mathcal K$ a choice of a dual path starting at $f$ that constitutes a line of disorder across 
which the sign of the coupling constant in~\eqref{def:ATdis} changes.
Again, when evaluating the correlations we project the nodes $\mathcal K$ onto the underlying corners. In particular, if a corner $c$ is doubled, then
$\psi_{c}=\psi_{c^<}=\psi_{c^>}$.
For every $S\in\pev(\mathcal K)$, we define ${F}^{S}$ to be the restriction of $F$ to the rows and columns indexed by $S$.
The following Pfaffian formula is derived from the lemma above and Proposition~\ref{lem:sousin}, as in the proof Corollary~\ref{cor:recPf}.
\begin{corollary}
Let $J$ be arbitrary and $U=0$. Then for all $S\in \pev(\mathcal K)$, 
\[
\langle \psi_S\rangle = \pf(F^S).
\]
\end{corollary}

We now turn our attention to determinants and state our last, and as far as we know, new result in the planar setup.
To this end, let again $\origin \cup \sink$ be a partition of $\mathcal K$ into sources and sinks with $|\origin|=|\sink|$ derived from the coloring $\mathcal B$.
We define the $\origin\times \sink$ square matrix $F^{\origin,\sink}$ by
\begin{align*}
F^{\origin,\sink}_{c,c'}= (-1)^{s(c,c')} \langle \psi_{c}\psi_{c'}\rangle  \qquad{c\in \origin, c'\in \sink},
\end{align*}
where $s(c,c')$ is the number of sources strictly between (the smaller and the larger corner) $c$ and $c'$ in the fixed order on $\mathcal K$.
The next theorem is a direct result of the corollary above and Proposition~\ref{lem:sousin}, as in the proof of Theorem~\ref{thm:varphi}.
\begin{theorem} \label{thm:psi}
Let $J$ be arbitrary and $U= 0$.Then
\begin{align*}
  \langle \chi_{A_1} \chi^2_{A_2} \tilde \chi_{B_1} \tilde \chi^2_{B_2} \rangle =\begin{cases}2^{-|\mathcal K|/2}\det (F^{\origin,\sink} )& \text{if } |\origin|=|\sink|, \\
0 &  \text{otherwise}.
\end{cases}
\end{align*}
\end{theorem}
Note that here, unlike for boundary spin correlations, we cannot conclude that the determinant above is nonnegative since some terms appearing in its random current expansion from Proposition~\ref{prop:nullexpansion} carry a negative sign.
We finish by remarking that the analogy from Remark~\ref{rem:fermion} with Dirac fermions applies also, for the same reasons, to the correlations of the $\chi$'s.

\section{Simon and Gaussian inequalities}
\label{sec:simon}
In the final section we discuss Ashkin-Teller spin correlations on arbitrary graphs.
We assume that $G=(V,E)$ is a finite, not necessarily planar graph, and
\begin{align} \label{eq:cond}
J_e\geq 0, \quad U_e\leq 0, \textnormal{ and } \cosh(2J_e)\geq e^{-2U_e},
\end{align}
for all $e\in E$.
Under these conditions Pfister and Velenik established in~\cite{PfiVel} that the spins $\sigma$ and $\tilde \sigma$ are negatively correlated. In particular, for $u,v,w\in V$, we have
\begin{align}\label{eq:PV}
\langle \sigma_u\sigma_w\rangle \langle \sigma_w\sigma_v\rangle \geq  \langle \sigma_u\sigma_w\tilde \sigma_w \tilde\sigma_v\rangle.
\end{align}
We first show a generalization of the original inequality of Simon~\cite{Simon} to the Ashkin--Teller model. 
We note that it is not clear to us if the improved inequality
due to Lieb~\cite{Lieb} is valid also in the interacting case $U<0$.
\begin{proposition}[Simon inequality] \label{thm:Simon} Let the coupling constants be as in~\eqref{eq:cond}, and let $u,v\in V$ be distinct vertices.
Let $W\subset V$ separate $u$ from $v$, i.e. $W$ is such that every path from $u$ to $v$ intersects $W$. Then
\begin{align*}
\langle \sigma_u\sigma_v\rangle \leq \sum_{w\in W} \langle \sigma_u\sigma_w\rangle \langle \sigma_w\sigma_v\rangle.
\end{align*}
\end{proposition}
\begin{proof}
Using \eqref{eq:PV} together with the switching lemma for $\sigma$ and $\tilde \sigma$, we can write
\begin{align*}
\sum_{w\in W} \frac{ \langle \sigma_u\sigma_w\rangle \langle \sigma_w\sigma_v\rangle}{ \langle \sigma_u\sigma_v\rangle} &\geq \sum_{w\in W} 
\frac{ \langle \sigma_u\sigma_w\tilde \sigma_w \tilde\sigma_v\rangle}{ \langle \sigma_u\sigma_v\rangle} \\
&= \frac{\sum_{\n \in \Omega_{uv}}  w_{\at}(\n)\big(\sum_{w\in W}\mathbf 1\{  u\con{\omega} w\}\big)}{\sum_{\n \in \Omega_{uv}}  w_{\at}(\n)} \\
&\geq 1.
\end{align*}
The second inequality holds true since $W$ separates $u$ from $v$, and since each current $\n=(\omega,\eta)\in \Omega_{uv}$ connects $u$ to $v$. This implies that 
\[
\sum_{w\in W}\mathbf 1\{  u\con{\omega} w\}\geq 1. \qedhere
\]
\end{proof}
We note that the above proof is the same as for the Ising model (see e.g.~\cite{RC}) once we have the inequality~\eqref{eq:PV}, and as the proofs in previous sections
it has a topological flavour. 
A standard consequence of the Simon inequality is the following sharpness statement (see e.g.~\cite[Corollary 9.38]{RC}).
\begin{corollary}
Let $\Gamma=(V,E)$ be an infinite, vertex-transitive graph, and let $J$ and $U$ be constant and as in~\eqref{eq:cond}. Denote by $\langle \cdot \rangle_{\Gamma}$ the expectation with respect to 
an infinite volume limit of the Ashkin--Teller model. 
Then if the susceptibility is finite, i.e.,
\[
\sum_{v\in V }\langle \sigma_o\sigma_v\rangle_{\Gamma}<\infty
\] for some vertex $o\in V$, then the correlations decay exponentially, i.e., 
there exists $C>0$ such that for all $v\in V$,
\[
\langle \sigma_o\sigma_v\rangle_{\Gamma} \leq e^{-Cd(o,v)},
\]
where $d$ is the graph distance.
\end{corollary}

The following inequality in the case of the Ising model was first established by Newman~\cite{newman}.
The name Gaussian comes from the fact that the inequality is saturated for Gaussian systems exactly as in the bosonic Wick's rule mentioned in Remark~\ref{rem:wick}.
\begin{proposition}[Gaussian inequality]
Let the coupling constants be as in~\eqref{eq:cond} and let $S\in \pev (V)$. Then
\[
\langle \sigma_S\rangle \leq \sum_{\pi \in \Pi(S)} \prod_{uv\in \pi} \langle \sigma_u \sigma_v\rangle.
\]
\end{proposition}

\begin{proof}
The inequality clearly holds true if $|S|=2$. We proceed by induction and assume that $|S|>2$.
We fix $v\in S$ and note that for every $\n\in \Omega_S$, we have
\[
1 \leq \mathop{\sum_{u \in S}}_{u\neq v} \mathbf 1\{ \n \in \mathfrak F_{uv} \}.
\]
Using the switching lemma for $\sigma$ and $\tilde \sigma$, and~\eqref{eq:PV} we get 
\[
\langle \sigma_S\rangle \leq  \mathop{\sum_{u \in S}}_{u\neq v}  \langle\tilde \sigma_u\tilde  \sigma_v \sigma_{S\setminus uv}\rangle\leq \mathop{\sum_{u \in S}}_{u\neq v} 
 \langle \sigma_u  \sigma_v\rangle \langle \sigma_{S\setminus uv}\rangle,
\]
and the desired inequality follows by applying the induction hypothesis to the sets $S\setminus uv$.
\end{proof}

We note that a special case when $|S|=4$ is the Lebowitz inequality and was first established for the Ashkin--Teller model in~\cite{ChSh} for the same range of coupling constants and using arguments
that are conceptually similar to ours.

\bibliographystyle{amsplain}
\bibliography{TotalPositivity1}
\end{document}